\documentclass[10pt]{article}
\usepackage{amssymb}
\usepackage{amsmath}
\usepackage{amsthm}
\usepackage{latexsym}
\usepackage[dvips]{epsfig}
\usepackage{mathrsfs}
\usepackage{eufrak}
\usepackage{yhmath}

\theoremstyle{plain}
\newtheorem{proposition}{Proposition}
\newtheorem{lemma}{Lemma}
\newtheorem{theorem}{Theorem}
\newtheorem{assumption}{Assumption}

\newtheorem{definition}{Definition}

\font\SYM=msbm10
\newcommand{\Real}{\mbox{\SYM R}}
\newcommand{\Complex}{\mbox{\SYM C}}

\newcommand{\Sphere}{\mbox{\SYM S}}

\font\tenscr=rsfs10 scaled1100
\font\sevenscr=rsfs7 % scaled \magstep1
\font\fivescr=rsfs5 % scaled \magstep1
\skewchar\tenscr='177
\skewchar\sevenscr='177
\skewchar\fivescr='177
\newfam\scrfam
\textfont\scrfam=\tenscr
\scriptfont\scrfam=\sevenscr
\scriptscriptfont\scrfam=\fivescr

\begin{document}

% QM 
%\bibliographystyle{/home/network/jav/tex/reporthack}
% Ludovica
%\bibliographystyle{/Users/Juan/Documents/tex/reporthack}

\title{\textbf{The ``non-Kerrness'' of domains of outer
    communication of black holes and exteriors of stars}}

\author{{\Large Thomas B\"ackdahl} \thanks{E-mail address:
{\tt t.backdahl@qmul.ac.uk}} \\
\vspace{5mm}
{\Large Juan A. Valiente Kroon} \thanks{E-mail address:
{\tt j.a.valiente-kroon@qmul.ac.uk}}\\
School of Mathematical Sciences,\\
 Queen Mary University of London, \\
Mile End Road, London E1 4NS, UK.}

\maketitle

\begin{abstract}
In this article we construct a geometric invariant for initial data
sets for the vacuum Einstein field equations
$(\mathcal{S},h_{ab},K_{ab})$, such that $\mathcal{S}$ is a
3-dimensional manifold with an asymptotically Euclidean end and an
inner boundary $\partial \mathcal{S}$ with the topology of the
2-sphere. The hypersurface $\mathcal{S}$ can be though of being in the
domain of outer communication of a black hole or in the exterior of a
star. The geometric invariant vanishes if and only if
$(\mathcal{S},h_{ab},K_{ab})$ is an initial data set for the Kerr
spacetime. The construction makes use of the notion of Killing spinors
and of an expression for a \emph{Killing spinor candidate} which can
be constructed out of concomitants of the Weyl tensor.
\end{abstract}


\section{Introduction}
Let $(\mathcal{S}, h_{ab}, K_{ab})$ be an initial
data set for the vacuum Einstein field equations such that
$\mathcal{S}$ has two asymptotically Euclidean ends, but otherwise
trivial topology\footnote{More precisely, $\mathcal{S} \approx
(\Real^3\setminus \mathcal{B}_1) \# (\Real^3\setminus \mathcal{B}_1)$
where $\mathcal{B}_1$ denotes an open ball of radius 1 and $\#$
indicates that the boundaries of the two copies of $(\Real^3\setminus
\mathcal{B}_1)$ are identified in the trivial way.}. In
\cite{BaeVal10a} a geometric invariant for this type of initial data
sets has been constructed ---see also \cite{BaeVal10b} for a detailed
discussion. This invariant is a non-negative number having the
property that it vanishes if and only if the initial data set
corresponds to data for the Kerr spacetime. Thus, the invariant
measures the \emph{non-Kerrness} of the initial
data.

\medskip
In view of possible applications of the non-Kerrness to the problem of
the uniqueness of stationary black holes and the non-linear stability
of the Kerr spacetimes a different type of initial hypersurface is of
more interest: a 3-dimensional hypersurface with the topology of the
complement of an open ball in $\Real^3$, $\mathcal{S} \approx
(\Real^3 \setminus \mathcal{B}_1)$. This type of 3-manifold
can be thought of as a Cauchy hypersurface in the domain of outer 
communication of a black hole or the exterior of a star. In the
present article we discuss the construction of a geometric invariant
measuring the non-Kerrness of this type of initial hypersurface.

\subsection*{Outline of the article}
In section \ref{Section:SummaryNonKerrness} we provide a brief
summary of the theory of non-Kerrness invariants developed in
\cite{BaeVal10a,BaeVal10b}. This is provided for quick reference and
contains the essential ingredients required in the
construction of the present article. Section
\ref{Section:KSCandidate} contains a discussion of properties of
vacuum Petrov type D spacetimes which are relevant for our
discussion. In particular, it provides a formula of a Killing spinor
candidate written entirely in terms of concomitants of the Weyl
tensor. For a spacetime that is exactly of Petrov type D, this
expression provides a Killing spinor of the spacetime. This expression
is used in the sequel to provide the boundary value of an elliptic
problem. Section \ref{Section:BoundaryValueProblem} provides a
discussion of a boundary value problem for the approximate Killing
spinor equation. Section \ref{Section:Invariant} makes use of the
solution to the boundary value problem to construct the non-Kerrness
invariant. Finally, in section \ref{Section:Conclusions} we provide
some conclusions and outlook. The article also includes two
appendices. The first one provides a summary of the results on
boundary value problems for elliptic systems used in our
construction. The second appendix contains an improved theorem
characterising the Kerr spacetime in terms of Killing spinors. This
theorem removes some technical assumptions made in
\cite{BaeVal10a,BaeVal10b}.

\subsection*{Notation}
All throughout, $(\mathcal{M},g_{\mu\nu})$ will denote an orientable
and time orientable, globally hyperbolic vacuum spacetime. It follows
that the spacetime admits a spin structure \cite{Ger68, Ger70c}. In
what follows, $\mu,\; \nu,\ldots$ will denote abstract 4-dimensional
tensor indices. The metric $g_{\mu\nu}$ will be taken to have
signature $(+,-,-,-)$. Let $\nabla_\mu$ denote the Levi-Civita
connection of $g_{\mu\nu}$. The triple $(\mathcal{S},h_{ab},K_{ab})$
will denote initial data on a hypersurface of the spacetime
$(\mathcal{M},g_{\mu\nu})$. The symmetric tensors $h_{ab}$, $K_{ab}$
will correspond, respectively, to the 3-metric and the extrinsic
curvature of the 3-manifold $\mathcal{S}$. The metric $h_{ab}$ will be
taken to be negative definite. The indices $a,\,b,\ldots$ will denote
abstract 3-dimensional tensor indices, while $i,\,j,\ldots$ will
denote 3-dimensional tensor coordinate indices.  Let $D_a$ denote the
Levi-Civita covariant derivative of $h_{ab}$. Spinors will be used
systematically. We follow the conventions of \cite{PenRin84}.  In
particular, $A,\,B,\ldots$ will denote abstract spinorial indices,
while $\mathbf{A}, \,\mathbf{B},\ldots$ will be indices with respect
to a specific frame.  Let $\nabla_{AA'}$ denote the spinorial
counterpart of the spacetime connection $\nabla_\mu$. Besides the
connection $\nabla_{AA'}$, two other spinorial connections will be
used: $D_{AB}$, the spinorial counterpart of the Levi-Civita covariant
derivative $D_a$ and $\nabla_{AB}$, spinorial version of the Sen
covariant derivative of $(\mathcal{S},h_{ab},K_{ab})$.

\section{Killing spinors and non-Kerrness}
\label{Section:SummaryNonKerrness}
In this section we provide a brief account of the theory of
non-Kerrness developed in \cite{BaeVal10a,BaeVal10b}.

\subsection{Killing spinors and Killing spinor initial data}
The starting point of the construction in \cite{BaeVal10a, BaeVal10b}
is the space-spinor decomposition of the Killing spinor equation
\begin{equation}
\nabla_{A'(A} \kappa_{BC)}=0,
\label{KillingSpinorEquation}
\end{equation}
where $\kappa_{AB}=\kappa_{(AB)}$ and the spinorial conventions of
\cite{PenRin84} are being used.

\medskip
Important for our purposes is the idea of how to encode that the
development of an initial data set $(\mathcal{S},h_{ab},K_{ab})$
admits a solution to the Killing spinor equation
\eqref{KillingSpinorEquation}. This question can be addressed by means
of a space-spinor formalism ---see e.g. \cite{Som80}, and moreover,
\cite{BaeVal10b} for a detailed account of the conventions being used.

\medskip
The space-spinor decomposition of equation
\eqref{KillingSpinorEquation} renders a set of 3 conditions intrinsic
to the hypersurface $\mathcal{S}$:
\begin{subequations}
\begin{eqnarray}
&& \xi_{ABCD}=0,\label{kspd1}\\
&& \Psi_{(ABC}{}^F\kappa_{D)F}=0, \label{kspd2}\\
&& 3\kappa_{(A}{}^E\nabla_B{}^F\Psi_{CD)EF}+\Psi_{(ABC}{}^F\xi_{D)F}=0,\label{kspd3}
\end{eqnarray}
\end{subequations}
where we have written
\[
\xi_{ABCD} \equiv \nabla_{(AB} \kappa_{CD)}, \quad \xi_{AB} \equiv
\frac{3}{2}\nabla_{(A}{}^{D}\kappa_{B)D}, \quad \xi \equiv \nabla^{PQ}\kappa_{PQ},
\]
and $\nabla_{AB}$ denotes the spinorial version of the Sen connection
associated to the pair $(h_{ab},K_{ab})$ of intrinsic metric and
extrinsic curvature. It can be expressed in terms of the spinorial
counterpart, $D_{AB}$ of the Levi-Civita connection of the 3-metric
$h_{ab}$, and the spinorial version, $K_{ABCD}=K_{(AB)(CD)}=K_{CDAB}$,
of the second fundamental form $K_{ab}$. For example, given a valence
1 spinor $\pi_{A}$ one has that
\[
\nabla_{AB}\pi_C = D_{AB} \pi_C +\tfrac{1}{2} K_{ABC}{}^Q \pi_Q,
\]  
with the obvious generalisations to higher valence spinors. In
equations \eqref{kspd2}-\eqref{kspd3}, the spinor $\Psi_{ABCD}$
denotes the restriction to the hypersurface $\mathcal{S}$ of the self-dual Weyl
 spinor. Crucially, the spinor $\Psi_{ABCD}$ can be written
entirely in terms of initial data quantities via the relations:
\[
\Psi_{ABCD} = E_{ABCD} + \mbox{i} B_{ABCD},
\]
with
\begin{eqnarray*}
&&  E_{ABCD}= -r_{(ABCD)} + \tfrac{1}{2}\Omega_{(AB}{}^{PQ}\Omega_{CD)PQ}
- \tfrac{1}{6}\Omega_{ABCD}K, \\
&&  B_{ABCD}=-\mbox{i}\ D^Q{}_{(A}\Omega_{BCD)Q},
\end{eqnarray*}
and where $\Omega_{ABCD}\equiv K_{(ABCD)}$, $K\equiv K_{PQ}{}^{PQ}$. Furthermore,
the spinor $r_{ABCD}$ is the Ricci tensor, $r_{ab}$, of the 3-metric
$h_{ab}$.

\medskip
The key property of the equations \eqref{kspd1}-\eqref{kspd3} is
contained in the following result proven in \cite{BaeVal10b} ---see
also \cite{GarVal08a}.

\begin{proposition}
Let equations \eqref{kspd1}-\eqref{kspd3} be satisfied for a symmetric
spinor $\check{\kappa}_{AB}$ on an open set $\mathcal{U}\subset
\mathcal{S}$, then the Killing spinor equation
\eqref{KillingSpinorEquation} has a solution, $\kappa_{AB}$, on the
future domain of dependence $\mathcal{D}^+(\mathcal{U})$.
\end{proposition}

\subsection{Approximate Killing spinors}

The \emph{spatial Killing spinor equation} \eqref{kspd1} can be
regarded as a (complex) generalisation of the conformal Killing vector 
equation. It will play a special role in our considerations. As in the
case of the conformal Killing equation, equation \eqref{kspd1} is
clearly overdetermined. However, one can construct a generalisation of
the equation which under suitable circumstances can always be expected
to have a solution.  One can do this by composing the operator in \eqref{kspd1} with
its formal adjoint ---see \cite{BaeVal10a}. This procedure renders the
equation
\begin{equation}
\mathbf{L}\kappa_{CD} \equiv \nabla^{AB} \nabla_{(AB} \kappa_{CD)}-\Omega^{ABF}{}_{(A}\nabla_{|DF|}\kappa_{B)C}-\Omega^{ABF}{}_{(A}\nabla_{B)F}\kappa_{CD}=0, \label{ApproximateKillingSpinorEquation} 
\end{equation}
which will be called the \emph{approximate Killing spinor
  equation}. One has the following result proved in \cite{BaeVal10b}:

\begin{lemma}
The operator $\mathbf{L}$ defined by the left hand side of equation
\eqref{ApproximateKillingSpinorEquation} is a formally self-adjoint elliptic
 operator.
\end{lemma}

In \cite{BaeVal10a,BaeVal10b} it has been shown that if
  $\mathcal{S}$ has the same topology as Cauchy slices of the Kerr
  spacetime, and if the pair $(h_{ab},K_{ab})$ is suitably
  asymptotically Euclidean, then there exists a certain asymptotic
  behaviour at infinity for the spinor $\kappa_{AB}$ for which the
  approximate Killing spinor equation always admits a solution.

\medskip
If one wants to extend the construction discussed in the previous
paragraphs to a 3-manifold on, say, the  domain of outer communication
of a black hole or the exterior of a star so that $\mathcal{S} \approx
(\Real^3\setminus \mathcal{B}_1)$, then in addition to prescribing the
asymptotic behaviour of the spinor $\kappa_{AB}$ at infinity, one also
has to prescribe the behaviour at the inner boundary $\partial
\mathcal{S}$. One wants to prescribe this information in such a way
that $\kappa_{AB}$ has the \emph{right} Killing behaviour at the
boundary whenever all of the Killing spinor data equations
\eqref{kspd1}-\eqref{kspd3} are satisfied. In this article we discuss
how this can be done, and as a result we construct the non-Kerrness for
3-manifolds with topology $(\Real^3\setminus \mathcal{B}_1)$. These
3-manifolds can be interpreted as slices in the domain of outer 
communication of a black hole or slices in the exterior of a
star. It is expected that this construction will be of use in the
reformulation of problems involving the Kerr spacetime: the uniqueness
of stationary black holes, the construction of an interior for the
Kerr solution, and possibly also the evolution of non-linear
perturbations of the Kerr spacetime.

\section{Petrov type D spacetimes}
\label{Section:KSCandidate}

In order to analyse what is the right initial data to be prescribed on
the boundary $\partial \mathcal{S}$ of our initial 3-manifold
$\mathcal{S}$, we will look at some properties of vacuum spacetimes of
Petrov type D.

\subsection{The canonical form for type D}
Let $\Psi_{ABCD}$ denote the Weyl spinor of a vacuum spacetime $(\mathcal{M},g_{\mu\nu})$. We shall consider the following invariants of $\Psi_{ABCD}$:
\begin{eqnarray*}
&& \mathcal{I}\equiv \tfrac{1}{2} \Psi_{ABCD} \Psi^{ABCD}, \\
&& \mathcal{J}\equiv \tfrac{1}{6} \Psi_{ABCD} \Psi^{CDEF} \Psi_{EF}{}^{AB}.
\end{eqnarray*}
The Petrov type of the spacetime is determined as a solution of the
eigenvalue problem
\[
\Psi_{ABCD}\eta^{CD} = \lambda \eta_{AB}
\]
---see e.g. \cite{SKMHH}. The eigenvalues $\lambda$ satisfy the equation
\begin{equation}
\label{polynomial1}
\lambda^3 -\mathcal{I}\lambda -2\mathcal{J}=0.
\end{equation}
Let $\lambda_1$, $\lambda_2$, $\lambda_3$ denote the roots of the
above polynomial. The invariants $\mathcal{I}$ and $\mathcal{J}$ can be expressed in terms
of the eigenvalues by
\begin{subequations}
\begin{eqnarray}
&& \mathcal{I} = \tfrac{1}{2}\left(\lambda_1^2 + \lambda_2^2 + \lambda_3^2  \right),\label{altI}\\
&& \mathcal{J}= \tfrac{1}{2} \lambda_1 \lambda_2 \lambda_3. \label{altJ}
\end{eqnarray}
\end{subequations}
In what follows we assume $\lambda_1,\;\lambda_2,\; \lambda_3 \neq
0$. The Petrov type D is characterised by the condition
$\lambda_1=\lambda_2$. Using expressions \eqref{altI}-\eqref{altJ},
one has that the remaining root satisfies the equation
\begin{equation}
\label{polynomial2}
\lambda^3_3 -2\mathcal{I} \lambda_3 +4 \mathcal{J}=0.
\end{equation}
Combining equations \eqref{polynomial1} and \eqref{polynomial2} one finds that
\[
\lambda_3 = 6\mathcal{J}\mathcal{I}^{-1}.
\] 

\medskip
For a Petrov type D spacetime there exist spinors (the principal
spinors) $\alpha_A$, $\beta_A$ satisfying the normalisation $\alpha_A
\beta^A=1$ such that
\begin{equation}
\Psi_{ABCD} = \psi \alpha_{(A}\alpha_B \beta_C \beta_{D)}, \quad
\psi=-\tfrac{1}{2}\lambda_3=- 3\mathcal{J}\mathcal{I}^{-1}.
\label{Psi}
\end{equation}
It will be convenient to define the spinor $\upsilon_{AB}\equiv
\alpha_{(A}\beta_{B)}$. Observe that because of our normalisation
conditions one has that $\upsilon_{AB}\upsilon^{AB}=-\tfrac{1}{2}$. Using the
spinor $\upsilon_{AB}$ one obtains the following alternative
expression for $\Psi_{ABCD}$:
\begin{equation}
\label{AltPsi}
\Psi_{ABCD} = \psi \left(\upsilon_{AB}\upsilon_{CD} +\tfrac{1}{6} h_{ABCD}  \right), \quad h_{ABCD} \equiv - \epsilon_{A(C} \epsilon_{D)B}. 
\end{equation}

\medskip
The expression \eqref{AltPsi} can be used to obtain a formula for the
spinor $\upsilon_{AB}$ in terms of the Weyl spinor $\Psi_{ABCD}$. Let
$\zeta_{AB}$ denote a non-vanishing symmetric spinor. Contracting
\eqref{AltPsi} with an arbitrary spinor $\zeta_{AB}$ one obtains:
\begin{subequations}
\begin{eqnarray}
&& \Psi_{ABCD}\zeta^{CD} = \psi \left(  \upsilon_{AB} \upsilon_{PQ}\zeta^{PQ} + \tfrac{1}{6}\zeta_{AB} \right), \label{Contraction:1}\\
&& \Psi_{ABCD}\zeta^{AB}\zeta^{CD} = \psi \left( (\upsilon_{PQ}\zeta^{PQ})^2 + \tfrac{1}{6} \zeta_{PQ}\zeta^{PQ} \right). \label{Contraction:2}
\end{eqnarray}
\end{subequations}
Using equation \eqref{Contraction:1} to solve for $\upsilon_{AB}$ and
equation \eqref{Contraction:2} to solve for $\upsilon_{PQ}\zeta^{PQ}$
one obtains the following formula for $\upsilon_{AB}$ in terms of
$\Psi_{ABCD}$ and the arbitrary spinor $\zeta_{AB}$:
\begin{equation}
\label{Precandidate}
\upsilon_{AB} = \Xi^{-1/2} \left(\psi^{-1}\Psi_{ABPQ} \zeta^{PQ} -\tfrac{1}{6}\zeta_{AB}  \right), 
\end{equation}
with
\begin{equation}
\Xi \equiv \psi^{-1}\Psi_{PQRS} \zeta^{PQ}\zeta^{RS} - \tfrac{1}{6} \zeta_{PQ}\zeta^{PQ}. \label{Xi}
\end{equation}
In the last formulae it is assumed that $\zeta_{AB}$ is chosen such that
\[
\Xi\neq 0.
\]

\subsection{The Killing spinor of a Petrov type D spacetime}

Let $\kappa_{AB}$ be a solution to the Killing spinor equation
\eqref{KillingSpinorEquation}. An important property of a Killing
spinor is that
\begin{equation}
\xi_{AA'} \equiv \nabla^Q{}_{A'} \kappa_{AQ}, \label{KillingVector1}
\end{equation}
satisfies the (spinorial version of the) Killing vector equation
\[
\nabla_{AA'} \xi_{BB'} + \nabla_{BB'}\xi_{AA'}=0.
\]
In general, the Killing vector $\xi_{AA'}$ given by formula
\eqref{KillingVector1} is complex ---that is, it encodes the
information of 2 real Killing vectors. This property is closely
related to the fact that all vacuum type D spacetimes admit, at least,
a pair of commuting Killing vectors ---see e.g. \cite{Kin69}. Vacuum
spacetimes of Petrov type D for which $\xi_{AA'}$ is real are 
called \emph{generalised Kerr-NUT spacetimes.}

\medskip
Every vacuum spacetime of Petrov type D has a Killing spinor ---see
\cite{PenRin86} and references therein. Indeed, in the notation of the
previous section, one has that
\begin{equation}
\kappa_{AB} = \psi^{-1/3} \upsilon_{AB}, 
\label{Formula:kappa}
\end{equation}
satisfies equation \eqref{KillingSpinorEquation}. Using formula
\eqref{Precandidate}, one obtains the following result:

\begin{proposition}
Let $(\mathcal{M},g_{\mu\nu})$ be a vacuum spacetime. If on
$\mathcal{U}\subset \mathcal{M}$,  the spacetime is of Petrov type D and
$\zeta_{AB}$ is a symmetric spinor satisfying
\[
\zeta_{AB}\neq 0, \quad \psi^{-1}\Psi_{PQRS} \zeta^{PQ}\zeta^{RS} - \tfrac{1}{6} \zeta_{PQ}\zeta^{PQ}\neq 0 \quad \mbox{ on } \mathcal{U},
\]
then
\begin{equation}
\kappa_{AB} = \psi^{-1/3} \Xi^{-1/2} \left(\psi^{-1}\Psi_{ABPQ} \zeta^{PQ} -\tfrac{1}{6}\zeta_{AB}  \right) \label{KillingSpinorFormula}
\end{equation}
with $\Xi$ given by \eqref{Xi} is a Killing spinor on
$\mathcal{U}$. The formula \eqref{KillingSpinorFormula} is independent
of the choice of $\zeta_{AB}$.
\end{proposition}

\medskip
That expression \eqref{Formula:kappa} is independent of the choice of
$\zeta_{AB}$ can be verified by writing 
\[
\zeta_{AB} = \zeta_0 \alpha_A \alpha_B+ \zeta_1 \alpha_{(A}\beta_{B)}
+ \zeta_2 \beta_A \beta_B,
\]
where $\{\alpha_A,\beta_A\}$ is the dyad given by equation
\eqref{Psi}.   Substituting the latter into \eqref{KillingSpinorFormula} one readily
obtains \eqref{Formula:kappa}.

\medskip
\noindent
\textbf{Observation.} Formula \eqref{KillingSpinorFormula} can be evaluated
for any vacuum spacetime $(\mathcal{M},g_{\mu\nu})$. In general, of
course, it will not give a solution to the Killing spinor equation
\eqref{KillingSpinorEquation}. The resulting spinor $\kappa_{AB}$
will depend upon the choice of $\zeta_{AB}$. We make the following definition:

\begin{definition}
\label{Definition:KSCandidate}
Let $(\mathcal{M},g_{\mu\nu})$ be a vacuum spacetime. Consider 
$\mathcal{U}\subset \mathcal{M}$ and on $\mathcal{U}$ a symmetric spinor
$\zeta_{AB}$  satisfying
\[
\zeta_{AB}\neq 0, \quad \psi^{-1}\Psi_{PQRS} \zeta^{PQ}\zeta^{RS} - \tfrac{1}{6} \zeta_{PQ}\zeta^{PQ}\neq 0 \quad \mbox{ on } \mathcal{U}.
\]
The symmetric spinor given by
\begin{equation}
\breve{\kappa}_{AB} = \psi^{-1/3} \Xi^{-1/2}
\left(\psi^{-1}\Psi_{ABPQ} \zeta^{PQ} -\tfrac{1}{6}\zeta_{AB}
\right), \label{Candidate}
\end{equation}
with
\[
 \Xi \equiv\psi^{-1}\Psi_{PQRS} \zeta^{PQ}\zeta^{RS} - \tfrac{1}{6}
\zeta_{PQ}\zeta^{PQ}, 
\]
will be called the $\zeta_{AB}$-Killing spinor candidate on $\mathcal{U}$.
\end{definition}

\medskip
\noindent
\textbf{Remark 1.} Although the choice of $\zeta_{AB}$ is essentially
arbitrary, as it will be seen, in many applications there is a natural
choice.

\medskip
\noindent
\textbf{Remark 2.} The choice of branch cut for the square root of
$\Xi$ can be chosen to be $\{ -r e^{i\theta}: r>0 \}$ where $\theta$
is the argument of $\psi^{-1}\Psi_{PQRS} \zeta^{PQ}\zeta^{RS} -
\tfrac{1}{6} \zeta_{PQ}\zeta^{PQ}$.

\section{A boundary value problem for the approximate Killing spinor equation} 
\label{Section:BoundaryValueProblem}

In this section we formulate a boundary value problem for the
approximate Killing spinor equation
\eqref{ApproximateKillingSpinorEquation} on a 3-manifold $\mathcal{S}
\approx \Real^3 \setminus \mathcal{B}_1$. As discussed in
the introduction, this type of 3-manifold can be thought of as a
Cauchy hypersurface in the domain of outer communication of a black
hole or the exterior of a star. For simplicity of the presentation, it
will be assumed that the initial data $(\mathcal{S},h_{ab},K_{ab})$
satisfies in its asymptotic region the behaviour:
\begin{subequations}
\begin{eqnarray}
&& h_{ij} = -\left(1+ \frac{2m}{r}\right)\delta_{ij} + o_\infty(r^{-3/2}), \label{Decay1} \\
&& K_{ij} = o_\infty(r^{-5/2}), \label{Decay2}
\end{eqnarray}
\end{subequations}
with $r=((x^1)^2 +(x^2)^2 + (x^3)^2)^{1/2}$, and $(x^1,x^2,x^3)$ are
asymptotically Cartesian coordinates. Our present discussion could be
extended at the expense of more technical details to include the case
of boosted initial data sets ---see e.g. \cite{BaeVal10b}. Here, and
in what follows, the fall off conditions of the various fields will be
expressed in terms of weighted Sobolev spaces $H^s_\beta$, where $s$
is a non-negative integer and $\beta$ is a real number. Here we use
the conventions for these spaces given in \cite{Bar86} ---see also
\cite{BaeVal10b}.  We say that $\eta\in H^\infty_\beta$ if $\eta\in
H^s_\beta$ for all $s$. Thus, the functions in $H^\infty_\beta$ are
smooth over $\mathcal{S}$ and have a fall off at infinity such that
$\partial^l \eta = o(r^{\beta-|l|})$. We will often write
$\eta=o_\infty(r^\beta)$ for $\eta\in H^\infty_\beta$ at the
asymptotic end.

\medskip
Following the ideas of \cite{BaeVal10a,BaeVal10b}, we shall look for
solutions to the approximate Killing spinor equation 
\eqref{ApproximateKillingSpinorEquation}  which expressed in terms of an
asymptotically Cartesian frame and coordinates have an asymptotic
behaviour given by
\begin{equation}
\label{Asymptotic:kappa}
\kappa_{\mathbf{AB}} = 
-\frac{\sqrt{2}}{3}\left (1+\frac{2m}{r}\right)x_{\mathbf{AB}}
+o_\infty(r^{-1/2}),
\end{equation}
with
\[
 \quad  x_{\mathbf{AB}} = \frac{1}{\sqrt{2}}
\left(
\begin{array}{cc}
-x^1 +\mbox{i}x^2 & x^3 \\
x^3 & x^1 +\mbox{i}x^2
\end{array}
\right).
\]

\subsection{Behaviour at the inner boundary}

The ideas of Section \ref{Section:KSCandidate} will be used to
prescribe the value of the spinor $\kappa_{AB}$ on the boundary $\partial
\mathcal{S}$. The configuration under consideration offers a natural
choice of spinor $\zeta_{AB}$ to evaluate the Killing spinor candidate
formula given in definition \ref{Definition:KSCandidate} ---namely,
the spinorial counterpart, $n_{AB}$, of the normal, $n_a$, to the
hypersurface $\partial \mathcal{S}$. By convention $n_a$ is assumed to
point outside $\mathcal{S}$ (outward pointing). Note that
because of the use of a negative definite 3-metric one has that
$n_{PQ}n^{PQ}=-1$. 

\medskip
 It will be convenient to define
the following set:
\[
\mathcal{Q} \equiv \{ z\in \Complex \;|\; z=\Xi(p), \;\; p \in \partial \mathcal{S}\},
\]
with
\[
\Xi =\psi^{-1}\Psi_{PQRS} n^{PQ}n^{RS} - \tfrac{1}{6}
n_{PQ}n^{PQ}. 
\]

\medskip
We shall make the following technical assumption on the initial data
set $(\mathcal{S}, h_{ab}, K_{ab})$:

\begin{assumption}
\label{TechnicalAssumption}
The initial data set $(\mathcal{S}, h_{ab}, K_{ab})$ is such that
$\Xi$ is a smooth function over $\partial \mathcal{S}$ satisfying 
\begin{itemize}
\item[(i)]  $0\not\in\mathcal{Q}$;
\item[(ii)] $\mathcal{Q}$ does not encircle
the the point $z=0$.
\end{itemize}
\end{assumption}

As a consequence of Assumption \ref{TechnicalAssumption} one can
choose a cut of the square root function on the complex plane such
that $\Xi^{-1/2}(p)$ is smooth for all $p\in \partial\mathcal{S}$.

% Since $\partial \mathcal{S}$ is compact and simply connected and $\Xi$
% is a smooth function, one has that $\mathcal{Q}$ is a bounded, simply
% connected domain of the complex plane. Consistent with definition
% \ref{Definition:KSCandidate}, we restrict our class of data to those
% for which $0\not \in \mathcal{Q}$. Under this assumption,
% $\mathcal{Q}$ cannot encircle the point $z=0$, and one can choose a
% cut of the square root on the complex plain such that $\Xi^{-1/2}(p)$
% is smooth for all $p\in\partial\mathcal{S}$. Accordingly, one has
% that:

% \begin{lemma}
% If $\Xi\neq 0$ on $\partial \mathcal{S}$, then the $n_{AB}$-Killing
% spinor candidate can be constructed such that it is a smooth spinorial
% field over $\partial \mathcal{S}$.
% \end{lemma}

\medskip
\noindent
\textbf{Remark 1.} The conditions in Assumption \ref{TechnicalAssumption} are satisfied by standard
Kerr data (in Boyer Lindquist coordinates) at the
horizon. Furthermore, by construction, the data
$(\mathcal{S},h_{ab},K_{ab})$ is data for the Kerr spacetime, then
the boundary data given by Killing Spinor candidate formula given by
definition \ref{Definition:KSCandidate} gives the right boundary
behaviour for the restriction of its Killing spinor to $\mathcal{S}$.

\medskip
\noindent
\textbf{Remark 2.} In order to match the asymptotic behaviour of the
$n_{AB}$-Killing spinor candidate given by
definition \ref{Definition:KSCandidate} with that given by equation
\eqref{Asymptotic:kappa} we add a normalisation factor to equation
\eqref{Candidate} to obtain
\begin{equation}
\breve{\kappa}'_{AB} = -\left(\frac{2^{5/6}m^{1/3}}{3^{1/6}}\right)\psi^{-1/3} \Xi^{-1/2}
\left(\psi^{-1}\Psi_{ABPQ} n^{PQ} -\tfrac{1}{6}n_{AB}
\right) \label{CandidateALT},
\end{equation}
where $m$ denotes the ADM mass of the asymptotic end. A direct
computation using the asymptotic expansions 
\[
\Psi_{\mathbf{ABCD}} = \frac{3m x_{(\mathbf{AB}} x_{\mathbf{CD})}}{r^5} + o_\infty(r^{-7/2}),
\quad \psi =\frac{6m}{r^3} +o_\infty(r^{-7/2}),
\]
justifies the extra normalisation. In these last expressions
$x_{\mathbf{AB}}$ points in the direction of the asymptotic end.

\subsection{Existence of solutions to the approximate Killing spinor equation}

Following the strategy put forward in \cite{BaeVal10a,BaeVal10b}, we
provide an Ansatz for a solution to the approximate Killing spinor
equation \eqref{ApproximateKillingSpinorEquation} which encodes the desired
behaviour at infinity. To this end, let 
\begin{equation}
\mathring{\kappa}_{\mathbf{AB}} \equiv -\frac{\sqrt{2}}{3}\left
  (1+\frac{2m}{r}\right)x_{\mathbf{AB}} \; \phi_R(r),
\label{LeadingBehaviour}
\end{equation}
where $\phi_R$ is a smooth cut-off function such that for $R>0$ large enough
\begin{eqnarray*}
&& \phi_R(r)=1, \quad  r\gg R, \\
&& \phi_R(r)=0, \quad r<R.
\end{eqnarray*}

\medskip
One then has the following result:

\begin{theorem}
\label{Theorem:ExistenceSolution}
Let $(\mathcal{S},h_{ab},K_{ab})$ be an initial data set for the
Einstein vacuum field equations such that $\mathcal{S}$ is a manifold
with a smooth boundary $\partial \mathcal{S}\approx \Sphere^2$
satisfying Assumption \ref{TechnicalAssumption}. Assume that
$(h_{ab},K_{ab})$ satisfy the asymptotic conditions
\eqref{Decay1}-\eqref{Decay2} with $m\neq 0$. Then, there exists a unique smooth
solution, $\kappa_{AB}$, to the approximate Killing equation
\eqref{ApproximateKillingSpinorEquation} with behaviour at the
asymptotic end of the form \eqref{Asymptotic:kappa} and with boundary
value at $\partial \mathcal{S}$ given by the $n_{AB}$-Killing spinor
candidate $\breve{\kappa}'_{AB}$ of equation \eqref{CandidateALT}.
\end{theorem}

\begin{proof}
Following the procedure described in \cite{BaeVal10a,BaeVal10b}, we
consider the Ansatz
\begin{equation}
\label{Ansatz}
\kappa_{AB}=\mathring{\kappa}_{AB} + \theta_{AB}, \quad \theta\in H^2_{-1/2}.
\end{equation}
 Substitution into equation \eqref{ApproximateKillingSpinorEquation} renders the following equation for the spinor
$\theta_{AB}$:
\begin{equation}
\label{elliptic:general}
\mathbf{L}\theta_{CD} = -\mathbf{L}\mathring{\kappa}_{CD}.
\end{equation}
In view that  $\mathring{\kappa}_{AB}$ vanishes outside the asymptotic
region, then the value of $\theta_{AB}$ at $\partial \mathcal{S}$
coincides with that of $\kappa_{AB}$. That is, we set
\begin{equation}
\label{elliptic:boundary}
\theta_{AB}|_{\partial \mathcal{S}}=\breve{\kappa}'_{AB}.
\end{equation}
By construction it follows that 
\[
\nabla_{(AB} \mathring{\kappa}_{CD)}\in
H^\infty_{-3/2},
\]
 so that 
\[
F_{CD}\equiv-\mathbf{L}\mathring{\kappa}_{CD}\in H^\infty_{-5/2}.
\]
The operator associated to the Dirichlet elliptic boundary value
problem \eqref{elliptic:general}-\eqref{elliptic:boundary} is given by
$(\mathbf{L},\mathbf{B})$ where $\mathbf{B}$ denotes the Dirichlet
boundary operator on $\partial \mathcal{S}$. As discussed in
\cite{BaeVal10a,BaeVal10b}, under assumptions
\eqref{Decay1}-\eqref{Decay2} the operator $\mathbf{L}$ is
asymptotically homogeneous ---see Appendix \ref{Appendix:Elliptic} for a
concise summary of the ideas and results of the theory elliptic
systems being used here. Now, elliptic boundary value problems with
Dirichlet boundary conditions satisfy the Lopatinski-Shapiro
compatibility conditions ---see \cite{WloRowLaw95}. Consequently, the
operator $(\mathbf{L},\mathbf{B})$ is L-elliptic and the map
\[
(\mathbf{L},\mathbf{B}): H^2_{-1/2}(\mathcal{S}) \rightarrow
H^0_{-5/2}(\mathcal{S}) \times H^{3/2}(\partial \mathcal{S})
\]
is Fredholm ---see theorem \ref{Fredholm:properties} of Appendix
\ref{Appendix:Elliptic}. The rest of the proof is an application of
the Fredholm alternative. Using Theorem \ref{Fredholm:Alternative}
with $\delta=-1/2$, one concludes that equation
\eqref{elliptic:general} has a unique solution if $F_{AB}$ is
orthogonal to all $\nu_{AB}\in H^0_{-1/2}$ in the Kernel of
$\mathbf{L}^*=\mathbf{L}$ with $\nu_{AB}=0$ on $\partial
\mathcal{S}$. If $\mathbf{L}\nu_{AB}=0$, then an integration by parts
shows that
\[
\int_{\mathcal{S}} \nabla^{(AB} \nu^{CD)}
\widehat{\nabla_{AB}\nu_{CD}}\mbox{d}\mu= \int_{\partial\mathcal{S}} n^{AB}\nu^{CD}\widehat{\nabla_{(AB}
  \nu_{CD)}} \mbox{d}S + 
\int_{\partial\mathcal{S}_\infty} n^{AB}\nu^{CD}\widehat{\nabla_{(AB}
  \nu_{CD)}} \mbox{d}S,
\]
where $\partial \mathcal{S}_\infty$ denotes the sphere at
infinity. The boundary integral over $\partial \mathcal{S}$ vanishes
because of $\nu_{AB}\in \mbox{Ker}(\mathbf{L},\mathbf{B})$, so that
$\nu_{AB}=0$ on $\partial \mathcal{S}$. As
$\nu_{AB}\in H^2_{-1/2}$ by assumption, it follows that $\nabla_{(AB}
\nu_{CD)} \in H^\infty_{-3/2}$ and furthermore that $n^{AB} \nu^{CD}
\widehat{\nabla_{(AB}\nu_{CD)}} = o(r^{-2})$.  An integral over a finite
sphere will then be of type $o(1)$. Thus, the integral over $\partial
S_\infty$ vanishes. Hence one concludes that
\[
\nabla_{(AB} \nu_{CD)}=0, \quad \mbox{ on }  \mathcal{S}.
\]
Using the same methods as in \cite{BaeVal10b},
Proposition 21 one finds that there are no non-trivial solutions to
the spatial Killing spinor equation that go to zero at infinity. 
Thus, there are no restrictions on $F_{AB}$ and equation
\eqref{elliptic:general} has a unique solution as desired. Due to
elliptic regularity, any $H^2_{-1/2}$ solution to equation
\eqref{elliptic:general} is in fact a $H^\infty_{-1/2}$ solution
---cfr. Lemma \ref{Lemma:EllipticRegularity}. Thus, $\theta_{AB}$ is smooth.
\end{proof}

\medskip
\noindent
\textbf{Remark.} It is worth mentioning that similar methods can be
used to obtain solutions to the approximate Killing spinor equation
equation on annular domains of the form $\mathcal{A} \equiv
\overline{\mathcal{B}_{R_2}\setminus \mathcal{B}_{R_1}}$, where
$R_2>R_1$. Again, one would use the Killing spinor candidate of
definition \ref{Definition:KSCandidate} to provide boundary value
data on the two components of $\partial \mathcal{A}$. This type of
construction is of potential relevance in the non-linear stability of
the Kerr spacetime and in the numerical evaluation of the non-Kerrness.

\section{The geometric invariant}
\label{Section:Invariant}

In this section we show how the approximate Killing spinor
$\kappa_{AB}$ obtained from Theorem \ref{Theorem:ExistenceSolution}
can be used to construct an invariant measuring the non-Kerrness of
the 3-manifold with boundary $\mathcal{S}$. To this end, we recall the
following lemma from \cite{BaeVal10a}:

\begin{lemma}
The approximate Killing spinor equation
\eqref{ApproximateKillingSpinorEquation}is the Euler-Lagrange equation
of the functional
\begin{equation}
J \equiv \int_{\mathcal{S}} \nabla_{(AB}\kappa_{CD)} \widehat{\nabla^{AB}\kappa^{CD}}\mbox{d}\mu.
\label{Functional}
\end{equation}
\end{lemma}
In what follows, it will be assumed that $\kappa_{AB}$ is the solution
to equation \eqref{ApproximateKillingSpinorEquation} given by
Theorem~\ref{Theorem:ExistenceSolution}. Furthermore, let
\begin{subequations}
\begin{eqnarray}
&& I_1 \equiv \int_{\mathcal{S}} \Psi_{(ABC}{}^{F}\kappa_{D)F}
\hat{\Psi}^{ABCG}\hat{\kappa}^D{}_G \mbox{d}\mu, \label{I1} \\
&& I_2 \equiv{} \int_{\mathcal{S}} \left(3\kappa_{(A}{}^{E}\nabla_{B}{}^{F}\Psi_{CD)EF}+\Psi_{(ABC}{}^{F}\xi_{D)F}\right) \nonumber \\
&& \hspace{3cm}\times \left(3\hat\kappa^{AP}\widehat{\nabla^{BQ}\Psi^{CD}{}_{PQ}}+\hat\Psi^{ABCP}\hat\xi^D{}_P\right){} \mbox{d}\mu. \label{I2}
\end{eqnarray}
\end{subequations}
The geometric invariant is then defined by
\begin{eqnarray}
I \equiv J + I_1 + I_2. \label{geometric:invariant}
\end{eqnarray}

\medskip
\noindent
\textbf{Remark.} It can be verified that $I$ is
coordinate independent. Furthermore, if the initial data set satisfies the 
decay conditions \eqref{Decay1}-\eqref{Decay2}, then $I$ is finite.

\medskip
The desired characterisation of Kerr data on 3-manifolds $\mathcal{S}$
with boundary and one asymptotic end is given by the following
theorem.

\begin{theorem}
\label{Theorem:CharacterisationInitialData}
Let $(\mathcal{S},h_{ab},K_{ab})$ be an initial data set for the
Einstein vacuum field equations such that $\mathcal{S}$ is a manifold
with boundary $\partial \mathcal{S}\approx \Sphere^2$ satisfying
Assumption \ref{TechnicalAssumption}.  Furthermore, assume that
$\mathcal{S}$ has only one asymptotic end, that the asymptotic
conditions \eqref{Decay1}-\eqref{Decay2} are satisfied with $m\neq 0$.
Let $I$ be the invariant defined by equations \eqref{Functional},
\eqref{I1}, \eqref{I2} and \eqref{geometric:invariant}, where
$\kappa_{AB}$ is given as the only solution to equation
\eqref{ApproximateKillingSpinorEquation} with asymptotic behaviour
given by \eqref{LeadingBehaviour} and with boundary value at $\partial
\mathcal{S}$ given by the $n_{AB}$-Killings spinor candidate
$\breve{\kappa}'_{AB}$ of equation \eqref{CandidateALT} where $n_{AB}$
is the outward pointing normal to $\partial\mathcal{S}$. The invariant
$I$ vanishes if and only if $(\mathcal{S},h_{ab},K_{ab})$ is an
initial data set for the Kerr spacetime.
\end{theorem}

The proof of this result is analogous to the one given in
\cite{BaeVal10a, BaeVal10b} and will be omitted.

\medskip
\noindent
\textbf{Remark.} In the previous theorem, for an initial data set for
the Kerr spacetime it will be understood that
$\mathcal{D}(\mathcal{S})$ (the union of the past and future domains
of dependence of $\mathcal{S}$) is isometric to a portion of the Kerr
spacetime. In order to make stronger assertions about
$\mathcal{D}(\mathcal{S})$, one needs to provide more information
about $\partial \mathcal{S}$. For example, if it can be asserted that
$\partial \mathcal{S}$ coincides with the intersection of the past and
future components of a non-expanding horizon, then as a consequence of
Theorem \ref{Theorem:CharacterisationInitialData} will give that
$\mathcal{D}(\mathcal{S})$ is the domain of outer communication of the
Kerr spacetime.

\section{Conclusions and outlook}
\label{Section:Conclusions}
Theorem \ref{Theorem:CharacterisationInitialData} and the methods
developed in the present articles are expected to be of relevance in
several outstanding problems concerning the Kerr spacetime: a proof of the
uniqueness of stationary black holes which does not assume analyticity
of the horizon, and whether the Kerr solution can describe the
exterior of a rotating star. The boundary value problem discussed in
the present article will also play a role in applications of Killing
spinor methods to the non-linear stability of the Kerr spacetime and
in the evaluation of the non-Kerrness in slices of numerically computed
black hole spacetimes.

\medskip
For the problem of the uniqueness of stationary black holes, as
mentioned in the remark after Theorem 
\ref{Theorem:CharacterisationInitialData}, one would like to consider
slices in the domain of outer communication of a stationary black hole
that intersect the intersection of the two components of the
non-expanding horizon. One then would have to analyse the consequences
that the existence of this type of boundary has on the Killing spinor
candidate constructed out of the normal to $\partial \mathcal{S}$
---the Weyl tensor is known to be of type D on non-expanding horizons
\cite{IonKla09a}. The main challenge in this approach is to find a
convenient way of relating the \emph{a priori} assumption about
stationarity made in the problem of uniqueness of black holes with the
\emph{Killing vector initial data candidates} $\xi$, $\xi_{AB}$
provided by the solution, $\kappa_{AB}$, to the approximate Killing
spinor equation \eqref{ApproximateKillingSpinorEquation}.

\medskip
With regards to the problem of the existence of an interior solution
for the Kerr spacetime, the key question to be analysed is what kind
of conditions on the boundary $\partial \mathcal{S}$ need to be
prescribed to ensure that the solution to the approximate Killing
spinor equation \eqref{ApproximateKillingSpinorEquation} given by
\eqref{Theorem:ExistenceSolution} renders a vanishing invariant $I$. It
is to be expected that these conditions will impose strong
restrictions to the type of matter models describing an hypothetical
interior solution.

\medskip
The issues touched upon in the previous paragraphs will be discussed
in future works.

\section*{Acknowledgements}
We would like to thank Gast\'on \'Avila for helpful conversations on
the boundary value problem for elliptic systems. TB is funded by a
scholarship of the Wenner-Gren foundations. JAVK was funded by an EPSRC
Advanced Research fellowship. The authors thank the hospitality and
financial support of the International Centre of Mathematical Sciences
(ICMS) and the Centre for Analysis and Partial Differential Equations
(CANPDE) of the University of Edinburgh for their hospitality during
the workshop on Mathematical Relativity September 1st-8th, 2010, in
the course of which 
this research was completed.

\appendix

\section{Elliptic results for slices in the domain of outer communication of a black hole}
\label{Appendix:Elliptic}
In this appendix we summarise the results on the theory of boundary
value problems for elliptic systems that have been used in the present
article. The presentation is adapted from \cite{LocMcO85}.

\medskip
As in the main text, let $\mathcal{S}$ denote a 3-dimensional manifold
with the topology of $\Real^3 \setminus \mathcal{B}_1$, where
$\mathcal{B}_1$ denotes the open ball of radius $1$. Note that
$\mathcal{S}$ is closed. Assume $\partial \mathcal{S}\approx
\Sphere^2$ to be $C^\infty$.  In what follows, let $u$ denote a
$N$-dimensional vector valued function over $\mathcal{S}$.  Following
\cite{Can81,Loc81}, a second order elliptic operator $\mathbf{L}$
acting on $u$ will be said to be asymptotically homogeneous if it
can be written in the form
\[
\mathbf{L} u(x) = (a_\infty^{ij}+ a^{ij}(x)) D_i D_j u(x) + a^i(x) D_i
u(x) + a(x) u(x), \quad x\in \mathcal{S} 
\]
where $a_\infty^{ij}$ denotes a matrix with constant coefficients
while $a^{ij}$, $a^{ij}$, $a$ are matrix valued functions of the
coordinates such that
\[
a^{ij} \in H^\infty_{-1/2}(\mathcal{S}), \quad a^j \in H^\infty_{-3/2}(\mathcal{S}), \quad a\in H^\infty_{-5/2}(\mathcal{S}).
\]

\medskip
On $\partial \mathcal{S}$ we will consider the homogeneous Dirichlet
operator $\mathbf{B}$ given by
\[
\mathbf{B}u(y) =u(y), \quad y\in \partial \mathcal{S}.
\]
The combined operator $(\mathbf{L},\mathbf{B})$ is said to be
\emph{L-elliptic} if $\mathbf{L}$ is elliptic on
$\mathcal{S}$ and $(\mathbf{L},\mathbf{B})$
satisfies the Lopatinski-Shapiro compatibility conditions ---see
\cite{WloRowLaw95} for detailed definitions. Crucial for our purposes
is that if $\mathbf{L}$ is elliptic and $\mathbf{B}$ is the Dirichlet
boundary operator, then the Lopatinski-Shapiro conditions are satisfied
and thus $(\mathbf{L},\mathbf{B})$ is L-elliptic ---see again
\cite{WloRowLaw95}, Theorem 10.7.

\medskip
The Fredholm properties for the combined operator
$(\mathbf{L},\mathbf{B})$ follow from Theorem 6.3 in \cite{LocMcO85}
---cfr. similar results in \cite{Kle99,Reu89}. Bartnik's conventions
are used for the weights of the Sobolev spaces $H^s_{\delta}$ ---see
\cite{Bar86}.

\begin{theorem}
\label{Fredholm:properties}
Let $\mathbf{L}$ denote a smooth second order asymptotically
homogeneous operator on $\mathcal{S}\approx \Real^3 \setminus
\mathcal{B}_1$. Furthermore, let $\partial \mathcal{S}$ be smooth and
let $\mathbf{B}$ denote the Dirichlet boundary operator. Then for
$\delta<0$, $s\geq 2$ the map
\[
(\mathbf{L},\mathbf{B}): H^s_\delta(\mathcal{S}) \rightarrow H^{s-2}_{\delta-2}(\mathcal{S}) \times H^{s-1/2}(\partial \mathcal{S}) 
\]
is Fredholm.
\end{theorem}

The same arguments used in Theorem 6.3 in \cite{Can81} then allow to prove
the following version of the Fredholm alternative:

\begin{proposition}
\label{Fredholm:Alternative}
Let $(\mathbf{L},\mathbf{B})$ as in theorem
\ref{Fredholm:properties}.  Given $\delta<0$, the boundary value problem
\begin{eqnarray*}
&& \mathbf{L} u(x) = f(x), \quad f\in H^0_{\delta-2}(\mathcal{S}) \quad x\in \mathcal{S}, \\ 
&&  u(y) = g(y), \quad g\in H^0(\partial\mathcal{S}), \quad y \in \partial \mathcal{S}
\end{eqnarray*}
 has a solution $u\in H^2_\delta(\mathcal{S})$ if
\[
\int_{\bar{\mathcal{S}}} f \cdot v \;\mbox{\emph d}\mu =0,
\]
for all $v\in H^0_{-1-\delta}(\mathcal{S})$ such that 
\begin{eqnarray*}
&& \mathbf{L}^* v(x)=0, \quad x\in \mathcal{S}, \\
&& v(y)=0, \quad y \in \partial\mathcal{S},
\end{eqnarray*}
where $\mathbf{L}^*$ denotes the formal adjoint of $\mathbf{L}$.
\end{proposition}

Finally, we note the following lemma ---cfr. equation (1.13) in
\cite{LocMcO85}.

\begin{lemma}
\label{Lemma:EllipticRegularity}
Let $(\mathbf{L},\mathbf{B})$ as in theorem
\ref{Fredholm:properties}. Then for any $\delta\in \Real$ and any
$s\geq 2$, there exists a constant $C$ such that for every
$u\in H^s_{loc}\cap H^0_\delta(\mathcal{S})$, the following inequality
holds
\[
||u||_{H^s_{\delta}(\mathcal{S})} \leq C\left( ||\mathbf{L}u||_{H^{s-2}_{\delta-2}(\mathcal{S})} + ||\mathbf{B}u||_{H^{s-1/2}(\partial\mathcal{S})} + ||u||_{H^{s-2}_\delta(\mathcal{S})} \right).
\]
\end{lemma}

In this lemma, $H^s_{loc}$ denotes the local Sobolev space. That is,
$u\in H^s_{loc}$ if for an arbitrary smooth function $v$ with compact
support, $uv\in H^s$.

\medskip
\noindent
\textbf{Remark.} If $\mathbf{L}$ has smooth coefficients and
$\mathbf{L}u=0$, then it follows that all the
$H^s_{\delta}(\mathcal{S})$ norms of $u$ are bounded by the
$H^0_\delta(\mathcal{S})$ and the $H^{s-1/2}(\partial \mathcal{S})$
norms. Thus, it follows that if a solution to the boundary value
problem exists and the boundary data is smooth, then the solution must
be, in fact, smooth ---elliptic regularity.

\section{An improved characterisation of the Kerr spacetime by means Killing spinors}

In \cite{BaeVal10b} a characterisation of the Kerr spacetime by means
of Killing spinors was given. This characterisation contains an 
\emph{a priori} assumption on the Weyl tensor ---namely, that it is
nowhere of type N or D. The purpose
of the present appendix is to show that these assumptions can be
removed.

\medskip
As in the main text, let $\kappa_{AB}$ denote a totally symmetric spinor. Let
\[
\xi_{AA'} \equiv \nabla^B{}_{A'} \kappa_{AB}.
\]
If $\kappa_{AB}$ is a solution to the Killing spinor equation
\eqref{KillingSpinorEquation}, then $\xi_{AA'}$ satisfies the Killing equation
\[
\nabla_{AA'} \xi_{BB'} + \nabla_{BB'}\xi_{AA'}=0.
\]
In general, $\xi_{AA'}$ will be a complex Killing vector. The Killing
form $F_{AA'BB'}$ associated to $\xi_{AA'}$ is defined by
\[
F_{AA'BB'} \equiv \tfrac{1}{2}\left( \nabla_{AA'} \xi_{BB'} - \nabla_{BB'} \xi_{AA'}   \right).
\]
In the cases where $\xi_{AA'}$ is real, we will consider the self-dual Killing form $\mathcal{F}_{AA'BB'}$ defined by
\[
\mathcal{F}_{AA'BB'} \equiv \tfrac{1}{2} (F_{AA'BB'} + \mbox{i} F^*_{AA'BB'}),
\]
where $F^*_{AA'BB'}$ is the Hodge dual of $F_{AA'BB'}$. Due to the symmetries of  of the self-dual Killing form one has that
\[
\mathcal{F}_{AA'BB'} = \mathcal{F}_{AB} \epsilon_{A'B'}, \quad \mathcal{F}_{AB} \equiv\tfrac{1}{2} F_{AQ'B}{}^{Q'} = \mathcal{F}_{BA}.
\]

\medskip
The characterisation of the Kerr spacetime discussed in
\cite{BaeVal10b} is, in turn, based on the following characterisation
proven by Mars \cite{Mar00}.

\begin{theorem}[Mars 1999, 2000]
\label{Theorem:Mars}
Let $(\mathcal{M},g_{\mu\nu})$ be a smooth vacuum spacetime with the
following properties:
\begin{itemize}
\item[(i)] $(\mathcal{M},g_{\mu\nu})$ admits a real Killing vector
$\xi_{AA'}$ such that the spinorial counterpart
of the Killing form of $\xi_{AA'}$ satisfies
\begin{equation}
\Psi_{ABCD} \mathcal{F}^{PQ} = \varphi \mathcal{F}_{AB},
\label{KeyProperty}
\end{equation}
with $\varphi$ a scalar;
\item[(ii)] $(\mathcal{M},g_{\mu\nu})$ contains a stationary
asymptotically flat 4-end, and $\xi_{AA'}$ tends to a time translation
at infinity and the Komar mass of the asymptotic end is non-zero.
\end{itemize}
Then $(\mathcal{M},g_{\mu\nu})$ is locally isometric to the Kerr spacetime.
\end{theorem}

\medskip
\noindent
\textbf{Remark.} A stationary asymptotically flat 4-end is an open
submanifold $\mathcal{M}_\infty\subset \mathcal{M}$ diffeomorphic to
to $I\times (\Real^3\setminus \mathcal{B}_R)$, where $I\subset \Real$
is an open interval and $\mathcal{B}_R$ is a closed ball of radius $R$
such that in local coordinates $(t,x^i)$ defined by the diffeomorphism
the metric satisfies
\[
|g_{\mu\nu}-\eta_{\mu\nu}| + |r\partial_i g_{\mu\nu}| \leq C r^{-\alpha}, \quad \partial_t g_{\mu\nu}=0,
\]
with $C$, $\alpha\geq 1$ constants, $\eta_{\mu\nu}$ the Minkowski metric and
\[
r = \sqrt{ \left(x^1\right)^2 + \left(x^2\right)^2 + \left(x^3\right)^2 }.
\]
In this context the notions of Komar and ADM mass coincide.

\bigskip
We want to relate the notion of Killing form and that of Killing
spinors. As discussed in \cite{BaeVal10b}, if $\xi_{AA'}$ is real, the
commutators for a vacuum spacetime readily yield that
\begin{equation}
\mathcal{F}_{AB}=\tfrac{3}{4} \Psi_{ABPQ} \kappa^{PQ}.
\label{KillingSpinorRelation}
\end{equation}
Now, vacuum spacetimes admitting a Killing spinor, $\kappa_{AB}$, can
only be of Petrov type D, N or O. If the spacetime is of type O at
some point (so that $\Psi_{ABCD}=0$), then
\eqref{KillingSpinorRelation} shows that $\mathcal{F}_{AB}=0$, and the
relation \eqref{KeyProperty} is satisfied trivially. If the spacetime
is of Petrov type N, then $\kappa_{AB}$ has a repeated principal
spinor which coincides with the repeated principal spinor of
$\Psi_{ABCD}$ ---see e.g. \cite{Jef84}. Hence, again one has that
$\mathcal{F}_{AB}=0$, and \eqref{KeyProperty} is satisfied
trivially. For Petrov type D spacetimes with a Killing spinor such
that $\xi_{AA'}$ is real,  it has already been shown in \cite{BaeVal10b} that
\eqref{KeyProperty} is satisfied.

\medskip
From the discussion in the previous paragraph, we obtain the following
characterisation of the Kerr spacetime in terms of Killing spinors.

\begin{theorem}
\label{Theorem:CharacterisationKerrData}
A smooth vacuum spacetime $(\mathcal{M},g_{\mu\nu})$ is locally
isometric to the Kerr spacetime if and only if the following conditions
are satisfied:
\begin{itemize}
\item[(i)] there exists a Killing spinor $\kappa_{AB}$ such that the
associated Killing vector $\xi_{AA'}$ is real;
\item[(ii)] the spacetime $(\mathcal{M},g_{\mu\nu})$ has a stationary
asymptotically flat 4-end with non-vanishing mass in which $\xi_{AA'}$
tends to a time translation.
\end{itemize}

\end{theorem}

As a consequence of this theorem, the \emph{a priori} conditions on
the Petrov type of the Weyl required in Theorem 28 of \cite{BaeVal10b}
can be dropped.

% Path in QM 
%\bibliography{/home/network/jav/tex/grbib}
% Path in Ludovica
%\bibliography{/Users/Juan/Documents/tex/grbib}

\end{document}